\documentclass[11pt]{article}
\usepackage[final]{epsfig}
\usepackage[left=3cm,right=3cm, top=2.5cm,bottom=2.5cm,bindingoffset=0cm]{geometry}
\usepackage{graphics}
\usepackage{amsmath}
\usepackage{amsfonts}
\usepackage{latexsym}
\usepackage{amssymb}
\usepackage{amsthm}
\usepackage{graphicx}
\usepackage{epstopdf}
\usepackage{multicol,multirow}
\usepackage{mathtools}

\usepackage{mathrsfs}

 \usepackage{relsize}

\usepackage{amsthm}

\usepackage[latin1]{inputenc}
\usepackage{tikz}
\usetikzlibrary{shapes,arrows,decorations.pathmorphing}
\usepackage{float}

\makeatletter
\def\thmhead@plain#1#2#3{%
  \thmname{#1}\thmnumber{\@ifnotempty{#1}{ }\@upn{#2}}%
  \thmnote{ {\the\thm@notefont#3}}}
\let\thmhead\thmhead@plain

\makeatother

\newcounter{AppCounter}

\def\restrict#1{\raise-.5ex\hbox{\ensuremath|}_{#1}}

\newtheorem{lemma}{Lemma}[section]
\newtheorem{proposition}[lemma]{Proposition}

\newtheorem{remark-definition}[lemma]{Remark-Definition}
\newtheorem{theorem}[lemma]{Theorem}
\newtheorem{corollary}[lemma]{Corollary}

\newtheorem{proposition-conjecture}[lemma]{Proposition-conjecture}

\theoremstyle{definition}

\newtheorem{definition}[lemma]{Definition}
\newtheorem{remark}[lemma]{Remark}

\newcommand{\marginnote}[1]
{
}

\newcounter{cy}

\newcounter{bk}

\newcounter{dps}

\usepackage{lipsum}

\title{A basis of Casimirs in 3D magnetohydrodynamics}

\author{ Boris Khesin\thanks{Department of Mathematics, University of Toronto, Toronto, ON M5S 2E4, Canada; e-mail:  \tt{khesin@math.toronto.edu}},
 Daniel Peralta-Salas\thanks{Instituto de Ciencias Matem\'{a}ticas, Consejo Superior de Investigaciones Cient\'{i}ficas,
28049 Madrid, Spain; e-mail:  \tt{dperalta@icmat.es}},
  and Cheng Yang\thanks{Department of Mathematics and Statistics, McMaster University, Hamilton, ON L8S 4K1, Canada, and the Fields Institute, Toronto, ON M5T 3J1, Canada;
 e-mail: \tt{yangc74@math.mcmaster.ca}
  }}

\date{January 2019}

\begin{document}

\maketitle
\begin{abstract}
We prove that any regular Casimir in 3D  magnetohydrodynamics is a function of
the magnetic helicity and cross-helicity. In other words, these two helicities
are the only independent regular integral invariants
of the coadjoint action of the MHD group $\text{SDiff}(M)\ltimes\mathfrak X^*(M)$,
which is the semidirect product of the group of volume-preserving diffeomorphisms and the dual space of its Lie algebra.
\end{abstract}


\section{Introduction} \label{intro}

The  motion  of an inviscid incompressible fluid on a closed Riemannian manifold $M$ is governed by the
classical Euler equation
\begin{equation}\label{idealEuler}
\partial_t v=-(v, \nabla) v-\nabla p\,,
\end{equation}
supplemented by the divergence-free condition ${\rm div}\, v=0$ on the velocity field $v$ of a fluid flow in $M$.
Here the term $(v, \nabla) v$ stands for the Riemannian covariant derivative
$\nabla_v v$ of the field $v$ along itself, and $p$ is the pressure function, which is uniquely defined up to an additive constant. This equation implies that the vorticity field $\omega=\text{curl}\, v$
is frozen into the fluid, a phenomenon that is known as Helmholtz's transport of vorticity. On a three-dimensional manifold $M$ endowed with the Riemannian volume form $d\mu$ this,
in turn, implies the conservation of helicity, a quantity that was discovered by Moreau~\cite{Mo61} and Moffatt~\cite{Mo69} in the 1960's:
$$
H(\omega, \omega):=\int_M \omega\cdot \text{curl}^{-1} \omega\,d\mu=\int_M \omega\cdot v\,d\mu\,.
$$

In~\cite{EnPeTo} it was proved that the helicity is the only $C^1$-Casimir of the 3D Euler equation; more precisely, any $C^1$-regular
functional of the vorticity that is invariant under the coadjoint action of the corresponding group of volume-preserving diffeomorphisms of $M$ must be a function of helicity. An analogous result in the context of three-manifolds with boundary, and divergence-free vector fields admitting a global cross section, was proved in~\cite{Ku14,Ku16}.

In this paper we describe a complete list of functionally independent Casimirs in self-consistent magnetohydrodynamics.
In MHD on a closed three-dimensional Riemannian manifold $M$
one considers an ideal incompressible fluid of infinite conductivity which carries a
magnetic field $B$. The field $B$ is transported by the fluid flow, i.e. it is frozen in it, and in turn reciprocally
acts (via the Lorenz force) on the conducting fluid. The corresponding equations of self-consistent magnetohydrodynamics
 described in Section~\ref{subsec:Group-theoretic} have two well-known first integrals discovered by Woltjer~\cite{Wo58}: the magnetic helicity of the field $B$, which is analogous to the hydrodynamic helicity defined above, and the cross-helicity, which is a measure of entanglement
of the fields $B$ and $\omega=\text{curl}\, v$. Our main theorem states that these two invariants are the only functionally independent MHD Casimirs,
i.e. invariants of the corresponding coadjoint action. More precisely, we show that any $C^1$-functional
that is invariant under the coadjoint action of the MHD group must be a function of magnetic helicity and cross-helicity. This extends (and actually recovers it as a particular case) the uniqueness of hydrodynamics helicity proved in~\cite{EnPeTo} to the context of MHD. The problem of finding a basis of Casimirs for the coadjoint orbits of the diffeomorphism group (or, more generally, the MHD group) is natural and was explicitly stated in~\cite[Section I.9]{arkh}.

The paper is organized as follows. In Section~\ref{sec:settings} we present a few facts about the MHD equations, including the invariance of cross-helicity and magnetic helicity under the coadjoint action of the MHD group (Sections~\ref{subsec:Group-theoretic} and~\ref{subsect:inv_helicity}) and we state the main theorem of this paper (Section~\ref{sec:main_result}). We divide the proof of the main result in four steps, which are presented in Section~\ref{sec:proof}. Finally, in Section~\ref{sec:appendix} we recall the Hamiltonian formulations of the Euler and MHD equations, as well as the explicit form of the corresponding coadjoint action in terms of the vorticity and magnetic fields.
\bigskip

\noindent {\bf Acknowledgments.}
B.K. was partially supported by an NSERC research grant.
D.P.-S. was supported by the ERC Starting Grant~335079, and partially supported by the ICMAT--Severo Ochoa grant SEV-2015-0554. A part of this work was done
while C.Y. was visiting  the Instituto de Ciencias Matem\'{a}ticas (ICMAT) in Spain. C.Y. is grateful to the ICMAT for its support and kind hospitality.


\section{Geometric Settings}\label{sec:settings}

\subsection{Equations of self-consistent magnetohydrodynamics}\label{subsec:Group-theoretic}

The evolution of an infinitely conducting ideal fluid carrying a magnetic field on
a closed three-dimensional Riemannian manifold $M$ is described by the following system of magnetohydrodynamics equations on the fluid velocity $v$ and the magnetic field $B$:
\begin{equation}\label{eq:mag_eq}
\left\{
  \begin{array}{l}

         \partial_t v = -(v,\nabla)v+(\text{curl}\; B)\times B-\nabla p\,,\\\\
	\partial_t B = -[v,B]\,,\\\\
		            \text{div}\; B=  \text{div}\; v=0\,.

\end{array} \right.
\end{equation}
Here $[v,B]$ stands for the Lie bracket of two vector fields, $v$ and $B$, and $\times$ denotes the cross product on the $3$-manifold.

Taking the curl on the first equation of~\eqref{eq:mag_eq}, we can rewrite the MHD equations as the evolution of the pair of fields $(\omega, B)$, where the field $\omega:={\rm curl}\,v$ is the vorticity field:
\begin{equation}\label{eq:magg}
         \partial_t \omega = [\omega,v]-[{\rm curl}\,B,B] \qquad \text{ and } \qquad \partial_t B = [B,v]\,.\\\\
\end{equation}

Consider the subspace $\mathfrak X(M)$ of exact divergence-free fields on $M$.
Recall that a divergence-free field $w$ is exact  if $w$ admits a field-potential, or, equivalently, if $i_wd\mu$ is an exact $2$-form. For example, on a closed three-dimensional manifold $M$ with trivial first  cohomology group, $H^1(M)=0$, all divergence-free fields are exact.
For an exact velocity field $v$, its evolution can be recovered from the vorticity evolution with the help of the ${\rm curl}^{-1}$-operator, since the curl operator on the space  of exact divergence-free fields on $M$ is one-to-one.

Furthermore, the $\rm curl$ operator on a Riemannian
manifold $M$ allows one to identify the space of exact divergence-free vector fields $\mathfrak X(M)$ and its dual $\mathfrak X^*(M)$,
as explained in Section~\ref{sec:appendix}.
It turns out that the MHD equations are Hamiltonian on the space of pairs
$(\omega, B)\in \mathfrak X(M)\times \mathfrak X(M)$.

\subsection{Invariance of the cross-helicity and magnetic helicity }\label{subsect:inv_helicity}

Consider the space of pairs $(\omega, B)\in \mathfrak X(M)\times \mathfrak X(M)$ of vorticity and magnetic fields on $M$.
\begin{definition}
The {\it magnetic helicity } is the following quadratic form on $B$:
$$
H(B,B):=\int_M B\cdot {\rm curl}^{-1} B\,d\mu\,.
$$
The  {\it cross-helicity } is the following bilinear form on $(\omega, B)$:
$$
H(\omega,B):=\int_M B\cdot {\rm curl}^{-1} \omega\,d\mu=\int_M B\cdot v\,d\mu\,,
$$
where $\omega ={\rm curl}\,v$ on $M$ (in other words, $v$ is the only field in $\mathfrak X(M)$ such that ${\rm curl}\,v=\omega$).
\end{definition}

These quantities turn out to be  invariant under the evolution of the MHD equations. Moreover,
they are invariant under the action of a group that generalizes the group $\text{SDiff}(M)$ of volume-preserving diffeomorphisms of the manifold $M$, similarly to the case of the hydrodynamics helicity. These properties are summarized in the following proposition.

\begin{proposition}\label{prop:invariance}
Both the magnetic helicity $H(B,B)$ and the cross-helicity $H(\omega,B)$
are first integrals of the MHD equations. Furthermore, they are Casimirs of the MHD equations, i.e. they are invariants
of the coadjoint action of the semidirect-product group $G={\rm SDiff}(M)\ltimes \mathfrak X^*(M)$ on its dual space
$\mathfrak g^*$.
\end{proposition}

\begin{proof}
It is a direct computation using the MHD equations and the explicit form of the coadjoint operator $\widetilde{\rm ad}^*$, see Section~\ref{sec:appendix} for details.
\end{proof}

\subsection{The main theorem}\label{sec:main_result}
As explained in Section~\ref{S:appmhd} (Remark~\ref{rem:action}), the coadjoint action $\widetilde{\rm ad}^*_{(v,A)}$ on $(\omega,B)\in \mathfrak X(M)\times \mathfrak X(M)$, expressed in terms of vector fields, is
\begin{equation}\label{eq:coadjoint}
 \widetilde{\rm ad}^*_{(v,A)}(\omega,B)=([\omega,v]-[A,B],[B,v])\,,
\end{equation}
where $(v,A)$ is any pair of divergence-free vector fields. Notice that the resulting fields $[\omega,v]-[A,B]$ and $[B,v]$ are exact. It is interesting to compare this expression with the MHD equations~\eqref{eq:magg}.

Our goal is to prove that, under appropriate regularity hypotheses, the magnetic helicity and cross-helicity form a basis of Casimirs of the aforementioned coadjoint action. To this end, we introduce the following definition, where we use $\mathfrak X^1(M)$ to denote the space of $C^1$ exact divergence-free fields on $M$. The $\widetilde{\rm ad}^*$ action on the space of smooth exact fields $\mathfrak X(M)\times\mathfrak X(M)$ naturally extends to $\mathfrak X^1(M)\times\mathfrak X^1(M)$.

\begin{definition}\label{def:regInt}
 Let $ F : \mathfrak X^1(M)\times\mathfrak X^1(M)\rightarrow\mathbb R$ be a $C^1$ functional.
  We say that $F$ is
a regular integral invariant if:

\noindent (i) It is invariant under the coadjoint action of the Lie group $ G=\text{SDiff}(M)\ltimes\mathfrak X^*(M)$, i.e., $F(\omega,B)= F(\widetilde{\rm Ad}^*_{\Phi}(\omega,B))$ for any $\Phi\in G$, and the group action $\widetilde{\text{Ad}}^*$ on $(\omega,B)$ is induced from the action $\widetilde{\rm ad}^*$ introduced in \eqref{eq:coadjoint}.

\noindent (ii) At any point $(\omega,B)\in\mathfrak X^1(M)\times\mathfrak X^1(M)$, the (Fr\'{e}chet) derivative of $F$ is an integral operator with
continuous kernel, that is,
\begin{equation}\label{eq:regInt}
( DF)_{(\omega,B)}(u,b)=\int_M K(\omega,B)\cdot(u,b)\,d\mu=\int_M (K_1(\omega,B)\cdot u+K_2(\omega,B)\cdot b)d\mu\,,
\end{equation}
for any $(u,b)\in\mathfrak X^1(M)\times\mathfrak X^1(M)$, where $ K=(K_1,K_2): \mathfrak X^1(M)\times\mathfrak X^1(M)\rightarrow\mathfrak X^1(M)\times\mathfrak X^1(M)$ is a continuous map. In this expression, the dot product denotes the scalar product of two vector fields using the Riemannian metric on $M$.
\end{definition}

It is easy to check that both the magnetic helicity and the cross-helicity are regular integral invariants in the sense of this definition. The following is the main result of this paper.

\begin{theorem}\label{thm:main}
Let $F$ be a regular integral invariant. Then $F$ is a function
of the magnetic helicity and cross-helicity,  i.e., there exists a $C^1$ function $f : \mathbb R\times\mathbb R \rightarrow \mathbb R$  such that $F(\omega,B) =f(H(B,B), H(\omega,B))$, where $(\omega,B)\in\mathfrak X^1(M)\times\mathfrak X^1(M)$.
\end{theorem}

\begin{corollary}\label{cor}
Let $F$ be a regular integral invariant that  depends  only on the magnetic field $B$.  Then $F$ is a function
of the magnetic helicity, i.e., there exists a $C^1$ function $g : \mathbb R \rightarrow \mathbb R$  such that $F(B) =g(H(B,B))$, where $B\in\mathfrak X^1(M)$.
\end{corollary}

\begin{remark}
By considering the subgroup $\text{SDiff}(M)\times \{0\}$ of $G$, one obtains the main theorem of \cite{EnPeTo}, which states that  any regular integral invariant of volume-preserving
transformations is a function of the helicity.
Indeed, the $G$-coadjoint action on the magnetic field $B$ coincides with the coadjoint action of $\text{SDiff}(M)$
on the vorticity, so Corollary~\ref{cor} on the magnetic helicity generalizes the corresponding result on the uniqueness of the
helicity invariant in ideal hydrodynamics.

The proof of Theorem \ref{thm:main} presented in Section~\ref{sec:proof} follows the  strategy of~\cite{EnPeTo}, but it is technically more involved since now we have two elements in the Casimir basis (the magnetic helicity and the cross-helicity), while the MHD semidirect group action is much more complicated.
\end{remark}

\begin{remark}
The use of the space $\mathfrak X^1(M)$ (endowed with the $C^1$ topology) is key in our proof of Theorem~\ref{thm:main}. 
The reason is that a main ingredient of the proof is Lemma~\ref{lem:const_firsInteg} below 
based on the existence of a residual subset of vector fields with special dynamical properties among exact divergence-free $C^1$ fields. 
The proof of the lemma makes use of a theorem by Mario Bessa~\cite{Besa} that is known to hold only for $C^1$ divergence-free 
vector fields with the $C^1$ topology, and is actually false for $C^4$ divergence-free vector fields with the $C^4$ topology due 
to the KAM theorem. A  result on the helicity uniqueness in the  $C^\infty$-setting, proved in~\cite{PeYa} for $C^\infty$ 
hydrodynamics using different tools from the theory of dynamical systems, allows one to similarly adjust  Lemma~\ref{lem:const_firsInteg} 
and generalize Theorem~\ref{thm:main} to $C^\infty$ magnetohydrodynamics as well.
\end{remark}

\section{Proof of the main theorem}\label{sec:proof}

\textbf{Step 1:} Consider a one-parameter family $\phi_t$ of elements on the semidirect product group $G$.
Let $F$ be a functional invariant under the coadjoint action of this family on the corresponding space of pairs $(\omega,B)$, i.e.
$$
F(\widetilde{\rm Ad}^*_{\phi_t}(\omega,B))=F(\omega,B)
$$
for all $t\in\mathbb R$. We assume that $\frac {d\phi_t}{dt}|_{t=0}=(v,A)$ is a pair of divergence-free vector fields and
$\phi_0={\rm id}$. Recall that $\mathfrak X(M)$ denotes the space of smooth exact divergence-free vector fields on $M$. 
For a pair of elements $(\omega,B)\in \mathfrak X(M)\times \mathfrak X(M)\subset \mathfrak X^1(M)\times \mathfrak X^1(M)$, 
one can take the time derivative of the expression above and evaluate it at $t=0$:

\begin{equation}
\begin{array}{rcl}
0&=&\frac {d}{dt}\Big|_{t=0}F(\widetilde{\rm Ad}^*_{\phi_t}(\omega,B))=(DF)_{(\omega,B)}(\widetilde{\rm ad}^*_{(v,A)}(\omega,B))\\\\
&=&\int_M \Big (K_1(\omega,B)\cdot ([\omega,v]-[A,B])-K_2(\omega,B)\cdot [v,B]\Big)d\mu\\\\
&=&-\int_M({\rm curl}\,K_1\times \omega+{\rm curl}\,K_2\times B)\cdot vd\mu-\int_M({\rm curl}\,K_1\times B)\cdot A \,d\mu\,.
\end{array}
\end{equation}

In the second line of the above computation we have used the definition (cf. Equation~\eqref{eq:regInt}) 
of the differential of a regular integral invariant and the expression~\eqref{eq:coadjoint} of the coadjoint 
action $\widetilde{\rm ad}^*_{(v,A)}$ on $(\omega,B)\in \mathfrak X(M)\times \mathfrak X(M)$, expressed in terms of vector fields. To pass to the third line we have used the identities that relate the commutator of divergence-free vector fields with the curl of the vector product, e.g. $[\omega,v]={\rm curl}\,(v\times\omega)$, and integrated by parts.

Since $v$ and $A$ are arbitrary divergence-free vector fields, the above computation shows that the vector fields ${\rm curl}\,K_1\times\omega+{\rm curl}\,K_2\times B$ and ${\rm curl}\,K_1\times B$ are $L^2$ orthogonal to all divergence-free vector fields on $M$. Then, by the Hodge decomposition theorem,  there exist two smooth
functions $P$ and $Q$ on $M$ such that
\begin{equation}\label{eq:first_Integral}
\left\{
  \begin{array}{l}

         {\rm curl}\,K_1\times\omega+ {\rm curl}\,K_2\times B=\nabla P\,,\\\\
		               {\rm curl}\,K_1\times B=\nabla Q \,.

\end{array} \right.
\end{equation}
This holds for all pairs $(\omega,B)\in \mathfrak X(M)\times \mathfrak X(M)$ of smooth exact divergence-free fields. Finally, the fact that the space $\mathfrak X(M)$ of smooth exact fields is an $L^2$ dense subset of the space $\mathfrak X^1(M)$ of $C^1$ exact fields, and the continuity of the functional, imply that for any pair $(\omega,B)\in\mathfrak X^1(M)\times\mathfrak X^1(M)$, there exists two $C^1$ functions $P$ and $Q$ on $M$ such that the equations~\eqref{eq:first_Integral} hold.

\medskip
\noindent \textbf{Step 2:}
Let us study the second equation in \eqref{eq:first_Integral} in more detail. Fix any field $\omega\in \mathfrak X^1(M)$ and any not identically zero field $B\in \mathfrak X^1(M)$. In what follows, we use the symbol $\mathfrak X^1_0(M)$ to denote the space of exact $C^1$ fields on $M$ that are not identically zero. The following lemma is the key result of this step; in its proof we will invoke Bessa's theorem~\cite{Besa}, where the use of $C^1$ divergence-free fields is essential.

\begin{lemma}\label{lem:const_firsInteg}
Assume that there are a $C^1$ function $J$ on $M$ and a  map $K:\mathfrak X^1(M)\times\mathfrak X^1(M)\rightarrow\mathfrak X^1(M)$
continuously  depending on the pair of fields  $(\omega,B)\in\mathfrak X^1(M)\times\mathfrak X^1_0(M)$ and
satisfying the equation
\begin{equation}\label{eq:crossproduct}
 {\rm curl}\,K(\omega,B)\times B=\nabla J\,.
\end{equation}
Then there is a  constant $C$, continuously depending on $(\omega,B)$, i.e.
there is  a continuous functional $C:\mathfrak X^1(M)\times\mathfrak X^1_0(M)\rightarrow\mathbb R$,  such that
\begin{equation}\label{eq:proportional}
 {\rm curl}\,K(\omega,B)=C(\omega,B)\;B
\end{equation}
on the whole manifold $M$.
\end{lemma}
\begin{proof}  From Equation~\eqref{eq:crossproduct}, we obtain $B\cdot\nabla J=0$, that is $J$ is a first integral of the vector field $B$. If $J$ is a constant on $M$ (a trivial first integral), then $\text{curl}\;K(\omega,B)\times B=0$, which implies that
$$\text{curl}\;K(\omega,B)=j(x)\;B\,,\;\;\text{with}\;j(x)=\frac{B\cdot K(\omega,B)}{|B|^2}\,,$$
for any point $x\in M\backslash  B^{-1}(0)$. It is apparent that the function $j\in C^0(M\backslash  B^{-1}(0))$ depends continuously on the pair $(\omega,B)$. Using a flow box argument as in the proof of~\cite[Step~2]{EnPeTo}, we conclude that the function $j$ is a continuous first integral of the vector field $B$ (in the complement of its zero set). Accordingly, for each $B\in\mathfrak X^1_0(M)$, either $J$ is a nontrivial $C^1$ first integral of $B$ on $M$ or $j$ is a $C^0$ first integral of $B$ in $M\setminus B^{-1}(0)$.

Furthermore, according to \cite{Besa}, there exists a residual (and hence dense)
set $\mathcal R$ of vector fields in $\mathfrak X^1(M)$ such that any $B \in \mathcal R$ is topologically transitive and
its zero set consists of finitely many hyperbolic points. Therefore, any continuous first integral of $B\in\mathcal R$ must be a constant, and so for any $B\in\mathcal R$ and any $\omega\in \mathfrak X^1(M)$, one has that $\text{curl}\;K(\omega,B)\times B=0$ on $M$. The assumption that the kernel $K$ is continuous, then implies that $\text{curl}\;K(\omega,B)\times B=0$ on $M$ for any $(\omega,B)\in\mathfrak X^1(M)\times \mathfrak X^1_0(M)$.

Now, we can define $j(x)\in C^0(M\setminus B^{-1}(0))$ as above such that $\text{curl}\;K(\omega,B)=j(x)\;B$ on $M\setminus B^{-1}(0)$. Moreover, arguing as before, one has that $\text{curl}\;K(\omega,B)=C(\omega,B)\;B$ on $M\backslash  B^{-1}(0)$, where $C(\omega,B)$ is a constant, for any pair $(\omega,B)\in \mathfrak X^1(M)\times \mathcal R$. Actually, since the zero set of $B$ consists of finitely many points, this identity holds on the whole of $M$. Since the map $j$ is continuous in $\mathfrak X^1(M)\times \mathfrak X^1_0(M)$, and is a constant $C(\omega, B)$ depending on $(\omega,B)$ on a dense subset of $\mathfrak X^1(M)\times \mathfrak X^1_0(M)$, it must also be a constant (depending on $(\omega,B)$) for all $(\omega,B)\in\mathfrak X^1(M)\times \mathfrak X^1_0(M)$, as we wanted to prove.
\end{proof}

\begin{remark}
This is a parametric version of the result proved in \cite[Step~3]{EnPeTo}, where a similar statement was shown for a vorticity field.
Here we prove it for a magnetic field $B$, regarding the vorticity $\omega$ as a parameter.
\end{remark}

\medskip
\noindent \textbf{Step 3:}
In this step we first use Lemma~\ref{lem:const_firsInteg} to show that ${\rm curl}\,K_1$ and ${\rm curl}\,K_2$ 
are linear combinations of $\omega$ and $B$ (with coefficients that are constants on $M$ depending on $\omega$ and $B$). 
After that, we complete the proof of Theorem~\ref{thm:main} assuming a property to be proved in Step~4.

\begin{lemma}\label{lem:derivative}
The kernel $K=(K_1,K_2)$ corresponding to the functional $F$ has the following property: there are constants $C_1$ and $C_2$ continuously depending on the fields $\omega\in \mathfrak X^1(M)$ and $B\in \mathfrak X^1_0(M)$, such that
$$
{\rm curl}\,K_1(\omega,B)=C_1(\omega,B)\;B
$$
and
$$
{\rm curl}\,K_2(\omega,B)=C_1(\omega,B)\;\omega+C_2(\omega,B)\;B
$$
for all $x\in M$.
\end{lemma}

\begin{proof}
First, applying Lemma~\ref{lem:const_firsInteg} to the second equation in~\eqref{eq:first_Integral}, we obtain that there exists a continuous functional $C_1:\mathfrak X^1(M)\times\mathfrak X^1_0(M)\rightarrow\mathbb R$  such that
$$
\text{curl}\;K_1(\omega,B)=C_1(\omega,B)\;B
$$
for all $x\in M$. Then, plugging this expression into the first equation of~\eqref{eq:first_Integral}, we get
$$
\text{curl}\;(K_2-C_1(\omega,B)\text{curl}^{-1}\omega)\times B=\nabla P\,,
$$
so using Lemma~\ref{lem:const_firsInteg} again, we obtain that there exists another continuous functional $C_2:\mathfrak X^1(M)\times\mathfrak X^1_0(M)\rightarrow\mathbb R$ such that
$$
\text{curl}\;(K_2-C_1(\omega,B)\text{curl}^{-1}\omega)= C_2(\omega,B)\;B\,,
$$
and therefore,
$$
\text{curl}\;K_2(\omega,B)=C_1(\omega,B)\;\omega+ C_2(\omega,B)\;B
$$
for all $x\in M$, as required.
\end{proof}

Finally, the main theorem follows from the formula for
the Fr\'{e}chet derivative of $F$ at the pair $(\omega,B)$ and the connectedness of common level sets of the helicity and cross-helicity, which is proved in the next step. Indeed, using Lemma~\ref{lem:derivative} we compute
the Fr\'{e}chet derivative of $F$ at $(\omega,B)\in \mathfrak X^1(M)\times \mathfrak X^1_0(M)$ as follows:

\begin{equation}\label{eq:derivative}
\begin{array}{l}
DF_{(\omega,B)}(u,b)=\int_M (K_1(\omega,B)\cdot u+K_2(\omega,B)\cdot b)\,d\mu\\\\
=\int_M C_1(\omega,B)(\text{curl}^{-1}u\cdot B+\text{curl}^{-1}\omega\cdot b)d\mu+\int_M C_2(\omega,B)\text{curl}^{-1}B\cdot b\,d\mu\\\\
=C_1(\omega,B)(DH)_{(\omega,B)}(u,b)+\frac 12 C_2(\omega,B)(DH)_{(B,B)}(b)\,.
\end{array}
\end{equation}
To pass to the second line we have used Lemma~\ref{lem:derivative} and integrated by parts. In the third line, we have substituted the expressions of the derivative for the cross-helicity and magnetic helicity:
\begin{align*}
(DH)_{(\omega,B)}(u,b)&=\int_M (\text{curl}^{-1}u\cdot B+\text{curl}^{-1}\omega\cdot b)d\mu\,,\\
(DH)_{(B,B)}(b)&=2\int_M \text{curl}^{-1}B\cdot b\,d\mu\,.
\end{align*}
Now the proof of the main theorem can be completed assuming the aforementioned connectedness of level sets of mixed helicity.
Namely, take any two pairs of vector fields at the same common level of the magnetic helicity and cross-helicity and connect them by
a path $(\omega_t, B_t)\in \mathfrak X^1(M)\times \mathfrak X^1_0(M)$. The differential of the functional $F$ along this path is given by
$$
(DF)_{(\omega,B)}(\dot\omega_t, \dot B_t)
=C_1(DH)_{(\omega,B)}(\dot\omega_t,\dot B_t)+ C_3(DH)_{(B,B)}(\dot B_t)\,,
$$
where $(\dot\omega_t, \dot B_t )$ are the tangent vectors along the path, and $C_3:=C_2/2$.
By choosing the path in such a way that the values $H(\omega_t,B_t)$ and $H(B_t,B_t)$ remain constant, the previous computation implies that $F$ is also constant on any common level set. Accordingly, there exists a function $f : \mathbb R\times\mathbb R\rightarrow\mathbb R$
which assigns a value of $F$ to each value of the mixed helicity, 
i.e., $F(\omega,B) =f(H(\omega,B),H(B,B))$ for all $(\omega,B)\in \mathfrak X^1(M)\times \mathfrak X^1_0(M)$. 
To include the case of identically zero $B$, we can take any pair $(\omega_0,B_0)\in \mathfrak X^1(M)\times \mathfrak X^1_0(M)$ such that $H(\omega_0,B_0)=H(B_0,B_0)=0$; since $H(\omega_0,(1-s)B_0)=H((1-s)B_0,(1-s)B_0)=0$ for $s\in[0,1]$, the continuity of the functional $F$ on $\mathfrak X^1(M)\times \mathfrak X^1(M)$ implies that $F(\omega,0)=f(0,0)$, and hence the property $F(\omega,B) =f(H(\omega,B),H(B,B))$ holds for all $(\omega,B)\in \mathfrak X^1(M)\times \mathfrak X^1(M)$. Additionally, $f$ is of class $C^1$ since $F$ itself is a $C^1$ functional.
The main theorem then follows once we prove Proposition~\ref{lem:pathconn} below on the connectedness of the level sets of the mixed helicity.

\medskip

\noindent \textbf{Step 4:}
 Define the mixed helicity as the $\Bbb R^2$-valued quadratic form on $\mathfrak X^1(M)$:
\begin{equation}\label{eq:mix_heli}
\begin{array}{rcl}
\mathfrak H:\mathfrak X^1(M)\times\mathfrak X^1(M)&\rightarrow& \mathbb R\times\mathbb R\\\
(\omega,B)&\mapsto&(H(B,B), H(\omega,B))\,.
\end{array}
\end{equation}

Our goal in this step is to prove that the (infinite-dimensional) level sets of the mixed helicity are path-connected:

\begin{proposition}\label{lem:pathconn}
The level sets of the mixed helicity $\mathfrak H$ are path-connected subsets of $\mathfrak X^1(M)\times\mathfrak X^1_0(M)$
(and hence of $\mathfrak X^1(M)\times\mathfrak X^1(M)$).
\end{proposition}

\begin{proof}
Let $(\omega_0,B_0)$ and $(\omega_1,B_1)$ be two pairs of vector fields in $\mathfrak X^1(M)\times\mathfrak X^1_0(M)\subset \mathfrak X^1(M)\times\mathfrak X^1(M)$ with the same mixed helicity, i.e.
$$
\mathfrak H(\omega_0,B_0)=\mathfrak H(\omega_1,B_1)=(a,b)\,.
$$
In order to prove the connectedness of the $(a,b)$-level set
we introduce two auxiliary vector fields $\xi\in \mathfrak X^1(M)$ and $\beta\in \mathfrak X^1_0(M)$ with the same value of mixed helicity, that is
$\mathfrak H(\xi,\beta)=(a,b)$, which can be connected with each pair $(\omega_i,B_i)$ by a path in $\mathfrak X^1(M)\times\mathfrak X^1_0(M)$ of constant mixed helicity. The only ingredient we need in the proof is the property that the curl operator acting on the space
$\mathfrak X^1(M)$ of exact fields has infinitely many positive and negative eigenvalues, which implies that the positive and negative subspaces of the helicity quadratic form $H(u,u)$
on $\mathfrak X^1(M)$ are infinite-dimensional.

To fix ideas, assume that both $a$ and $b$ are positive (other signs and the cases of vanishing  $a$ or $b$ are treated similarly).
Consider the subspace $\mathcal S\subset \mathfrak X^1(M)$ of vector fields orthogonal to
the four vector fields $\omega_0,\omega_1,B_0,B_1$ with respect to the helicity quadratic form $H$, that is
$$
\mathcal S:=\{u\in \mathfrak X^1(M)~|~ \int_M u\cdot \text{curl}^{-1}\omega_kd\mu=\int_M u\cdot \text{curl}^{-1}B_kd\mu=0 
~~\text{for}~k=0,1\}\,.
$$
This space has codimension $\le 4$ and hence the restriction of the helicity $H$ to this subspace is still sign-indefinite. Hence,
one can choose a field $\beta\in \mathcal S$ such that $H(\beta,\beta)=1$ (for other signs of $a$ and $b$
one needs to choose  $H(\beta,\beta)=-1$ or  $H(\beta,\beta)=0$ for some non-zero field $\beta$).

Now define a family of vector fields $B_t:= (1-t)B_0+f(t)\beta$ for $t\in [0,1]$, and choose an appropriate function $f(t)$
so that the condition $H(B_t,B_t)=a$ holds for all $t$. Namely,
$$
H(B_t,B_t)=H\Big((1-t)B_0+f(t)\beta,(1-t)B_0+f(t)\beta\Big)
$$
$$
=(1-t)^2H(B_0,B_0)+f(t)^2H(\beta,\beta)= (1-t)^2a+f(t)^2\,,
$$
where we have used that $H(B_0,B_0)=a$, $H(B_0,\beta)=H(\beta,B_0)=0$ and $H(\beta,\beta)=1$.
Then, taking $f(t):=\sqrt{(2t-t^2)a}$ for $t\in [0,1]$ we obtain a continuous family $B_t$ of fields in $\mathfrak X^1_0(M)$
that have constant helicity $a$ and connect $B_0$ and $\sqrt a\beta$. In the same way one can connect $\sqrt a\beta$ and $B_1$ for $t\in [1,2]$.

Now define a family of fields $\omega_t:=(1-t)\omega_0+ g(t)\beta$ for a function $g(t)$ with $t\in [0,1]$, starting at
$\omega_0$ and ending at
$\xi:=g(1)\beta$. The function $g(t)$ is to be chosen such that $H(\omega_t, B_t)=b$ for all $t\in [0,1]$. The condition we obtain is
$$
H(\omega_t,B_t)=H\Big((1-t)\omega_0+g(t)\beta,(1-t)B_0+f(t)\beta\Big)
$$
$$
=(1-t)^2H(\omega_0,B_0)+f(t)g(t)H(\beta,\beta)= (1-t)^2b+f(t)g(t)\,,
$$
where we have used that $H(\omega_0,B_0)=b$, $H(\beta,\beta)=1$ and $H(\omega_0,\beta)=H(\beta,B_0)=0$.
Then, taking $g(t):=\sqrt{{(2t-t^2)b^2}/{a}}$, we obtain a continuous family of fields $\omega_t\in \mathfrak X^1(M)$ that have the cross-helicity with $B_t$ independent of $t$ and equal to $b$, and connect $\omega_0$ with $\xi=b/\sqrt a\beta$. Similarly, we can connect $\xi$ with $\omega_1$ for $t\in [1,2]$.
We have hence shown that the level set $\mathfrak H^{-1}(a,b)$ of mixed helicity is path-connected. The connecting path can be smoothened out by adjusting this construction.

The cases of one or both $a$ and $b$ vanishing are analogous, but one may need to take vector fields $\beta$ and $\xi$ in $\mathcal S$ that are not proportional to each other. For example, for $a=b=0$, one can take linearly independent $\beta$ and $\xi$ such that $H(\beta,\beta)=H(\beta,\xi)=0$, and the following families of vector fields: $B_t=(1-t)B_0+t\beta$ and $\omega_t=(1-t)\omega_0+t\xi$, for $t\in[0,1]$ (and analogously for $(\omega_1,B_1)$).
Such fields $\beta,\xi$ exist because both positive and negative subspaces of the helicity
quadratic form are of dimension greater than 4 (in fact, infinite-dimensional).
The constructed path $(B_t, \omega_t)$ for $a=b=0$ proves the connectedness of the level sets in $\mathfrak X^1(M)\times\mathfrak X^1_0(M)$ (since $H(\omega,0)=0$ for all $\omega\in\mathfrak X^1(M)$), as required.
\end{proof}


\section{Appendix: Hamiltonian formulation of the MHD equations}\label{sec:appendix}

\subsection{Hamiltonian setting of ideal hydrodynamics}
The  Euler equation $\partial_t v=-(v, \nabla) v-\nabla p$
 of ideal hydrodynamics can be regarded as an
equation of the geodesic flow on the group $\text{SDiff}(M)$
of volume-preserving diffeomorphisms of $M$ with respect to the right-invariant metric on the group
given by the $L^2$-energy of the velocity field. In these terms, the Euler equation
describes an evolution in the Lie algebra $\mathfrak X(M)$ of divergence-free vector fields on $M$,
tracing the geodesics in the group $\text{SDiff}(M)$.

This point of view implies the following Hamiltonian
reformulation of the Euler equation.
Consider the (regular) dual space $\mathfrak X^*(M)$
to the Lie algebra $\mathfrak X(M)$. This dual space $\mathfrak X^*(M)$ has a natural  description
as the space of cosets  $\mathfrak X^*(M)=\Omega^1(M) / d \Omega^0(M)$ of 1-forms modulo exact 1-forms on $M$,
where the coadjoint action of the group $\text{SDiff}(M)$ on the dual
 $\mathfrak X^*(M)$ is given by the change of coordinates in (cosets of) 1-forms on $M$
 by means of volume-preserving diffeomorphisms, see~\cite{arkh}.

Recall that the manifold $M$ is equipped with a Riemannian metric $(\cdot,\cdot)$, and it allows one to
identify the Lie algebra and its dual by means of the so-called inertia operator $\mathbb I: \mathfrak X(M)\to\mathfrak X^*(M)$.
Namely, given a vector field $v$ on $M$  one defines the 1-form $u=v^\flat$
as the pointwise inner product with vectors of the velocity field $v$:
$v^\flat(W): = (v,W)$ for all $W\in T_xM$, see details in \cite{arkh}.
The Euler equation \eqref{idealEuler} rewritten on 1-forms is $\partial_t u=-L_v u-dP$
for the 1-form $u=v^\flat$ and an appropriate function $P$ on $M$.
In terms of the cosets of 1-forms $[u]=\{u+df\,|\,f\in C^\infty(M)\}\in \Omega^1(M) / d \Omega^0(M)$, the Euler equation looks as follows:
\begin{equation}\label{1-forms}
\partial_t [u]=-L_v [u]
\end{equation}
on the dual space $\mathfrak X^*(M)$, where $L_v$ is the Lie derivative along the field $v$.
The Euler equation is the Hamiltonian equation on the dual space $\mathfrak X^*(M)$ with respect to
the Lie-Poisson structure and with the Hamiltonian functional
$$
E ([u] ):=\frac 12 \langle [u], \mathbb I^{-1} [u]\rangle=\frac 12\int_M u(v)\,d\mu=\frac 12\int_M v\cdot v\,d\mu
$$
for $u=v^\flat$,
given by the kinetic energy of the fluid, see details in~\cite{arkh}.
The corresponding Hamiltonian operator is given by the Lie algebra coadjoint action ${\rm ad}^*_v$,
which coincides with the Lie derivative in the case of the diffeomorphisms group:
$$
{\rm ad}^*_v [u] =L_v[u]\,.
$$
Its symplectic leaves are coadjoint orbits of the corresponding group $\text{SDiff}(M)$.

Furthermore, one can introduce the vorticity 2-form $\zeta:=du$ as the differential of the 1-form
$u=v^\flat$. The vorticity exact 2-form is well-defined for cosets $[u]$:
1-forms $u$ in the same coset have equal vorticities $\zeta=du$.
The corresponding Euler equation assumes the vorticity (or Helmholtz) form
\begin{equation}\label{idealvorticity}
\partial_t \zeta=-L_v \zeta\,,
\end{equation}
which means that the vorticity form is transported by (or ``frozen into") the fluid flow.

In 3D the vorticity 2-form $\zeta$ can be identified with the (divergence-free and exact) vorticity vector field $\omega={\rm curl}\,v$ by means of the volume form $d\mu$ on  $M$: $i_\omega d\mu=\zeta$. The corresponding Euler evolution of the vorticity field
is given by the same transport equation: $\partial_t \omega=-L_v \omega$.
The helicity
$$
H(\omega, \omega)=\int_M \omega\cdot \text{curl}^{-1} \omega\,d\mu=\int_M \omega\cdot v\,d\mu
$$
of the field $\omega$ is a Casimir (i.e. an invariant of the coadjoint action) on the dual space $\mathfrak X^*(M)$, and hence a first integral of the Euler equation.

\subsection{Hamiltonian setting of ideal MHD}\label{S:appmhd}

It turns out that the MHD equations (\ref{eq:mag_eq}) can be studied in the same manner \cite{ViDo}.
These equations are related to the semidirect-product group $G=\text{SDiff}(M)\ltimes \mathfrak X^*(M)$ of the volume-preserving diffeomorphisms group $\text{SDiff}(M)$
and the dual space $\mathfrak X^*(M)=\Omega^1(M)/d\Omega^0(M)$ of the Lie algebra $\mathfrak X(M)$
of divergence-free vector fields.
Its Lie algebra is ${\mathfrak g}=\mathfrak X(M)\ltimes\mathfrak X^*(M)$
and the corresponding dual space is
$$
{\mathfrak g}^*=\mathfrak X^*(M)\oplus\mathfrak X(M)=\Omega^1(M)/d\Omega^0(M)\oplus\mathfrak X(M)\,.
$$
The coadjoint Lie-algebra ${\mathfrak g}$-action on its dual ${\mathfrak g}^*$ is given by
\begin{equation}\label{eq:coadj_action1}
\text{ad}^*_{(v,[\alpha])}([u],B)=(L_v[u]-L_B[\alpha],-[v,B])\,,
\end{equation}
where $([u],B)\in {\mathfrak g}^*=\Omega^1(M)/d\Omega^0(M)\oplus\mathfrak X(M)$ and $(v,[\alpha])\in {\mathfrak g}=\mathfrak X(M)\ltimes\Omega^1(M)/d\Omega^0(M)$.


Similarly to ideal hydrodynamics, one can regard the MHD equations as the equations of the geodesic flow on
the semidirect product group $G=\text{SDiff}(M)\ltimes \mathfrak X^*(M)$. It has the following Hamiltonian form
on the dual space $\mathfrak g^*$.
The Hamiltonian function is the following  quadratic energy
$$
E([u],B)=\frac 12\langle [u],\mathbb I^{-1}[u]\rangle+\frac 12\langle B,\mathbb I B\rangle
$$
on the dual space  ${\mathfrak g}^*$ with the Lie-Poisson structure.
Here the map $\mathbb I: \mathfrak X(M)\to \mathfrak X^*(M)=\Omega^1(M)/d\Omega^0(M)$
is the inertia operator from the (non-extended) Lie algebra  $\mathfrak X(M)$ of divergence-free vector fields to its dual.
The MHD equations can be written  on the dual space ${\mathfrak g}^*$ as follows:
\begin{equation}\label{eq:mag_eq2}
\left\{
  \begin{array}{l}
         \partial_t{[u]} = -L_v[u]+L_B[b],\\\\
		               \partial_t B = -[v,B],
\end{array} \right.
\end{equation}
where $v=\mathbb I^{-1}[u]\in \mathfrak X(M) $ and $[b]=\mathbb I B\in\mathfrak X^*(M)$.
We refer to \cite{arkh} and \cite{khch} for more details.

\begin{remark}
By taking the differential of the first equation in~\eqref{eq:mag_eq2} and recalling the definition of the vorticity two-form $\zeta=d[u]$, we obtain the equation
$\partial_t \zeta =-L_v\zeta +L_B\,d[b]$, i.e. this equation manifests that the vorticity $\zeta$, as well as
the vorticity field $\omega={\rm curl }\,v$  defined by $\zeta=i_\omega d\mu$, is not frozen into the flow, but differs from
the would-be-transported one by a term depending on $B$. Nevertheless, the  cross-helicity $H(\omega,B)$
of  exact fields $\omega$ and $B$ is conserved,
as well as the helicity $H(B,B)$ of the transported magnetic field $B$.
\end{remark}

\begin{remark}\label{rem:action}
Now we compute the coadjoint action in terms of the vector fields $(\omega, B)$, rather than the pairs $([u], B)$, involving cosets of 1-forms.
For simplicity in the exposition, we assume that the closed three-dimensional Riemannian manifold $M$ has trivial cohomology
$H^1(M)=0$.
To pass between the vorticity fields $\omega={\rm curl}\,v$ and the corresponding cosets $[u]$ of 1-forms $u=v^\flat$
 we  introduce the operator $\sigma: \omega \mapsto [u]$ defined by $u:=({\rm curl}^{-1}\omega)^\flat$, i.e.
$\sigma = \mathbb I\circ {\rm curl}^{-1}$. Note that although both $\mathbb I$ and ${\rm curl}^{-1}$ are metric dependent, the operator $\sigma$ depends on the volume form $d\mu$ only, since $\omega$ is the kernel of $d[u]$, i.e. $i_{\omega}d\mu=d[u]=d\,\sigma(\omega)$.

Using this operator $\sigma$, which is an isomorphism between $\mathfrak X(M)$ and $\mathfrak X^*(M)$, we have the following space identification,
$$
\mathfrak g=\mathfrak X(M)\ltimes \mathfrak X^*(M)\simeq\mathfrak X(M)\ltimes \sigma^{-1}(\mathfrak X^*(M))=\mathfrak X(M)\times\mathfrak X(M),
$$
and
$$
\mathfrak g^*=\mathfrak X^*(M)\oplus \mathfrak X(M) \simeq \sigma^{-1}(\mathfrak X^*(M))\oplus \mathfrak X(M)=\mathfrak X(M)\times\mathfrak X(M).
$$
The natural pairing $\langle\cdot,\cdot\rangle$ between $\mathfrak X(M)$ and $\mathfrak X^*(M)$ becomes
$$
\langle w,v\rangle=\langle [u],v\rangle=\int_M \text{curl}^{-1}w\cdot v\;d\mu\,,
$$
where $w\in\mathfrak X(M)$, $v\in\mathfrak X(M)$ and $[u]=\sigma(w)\in\mathfrak X^*(M)$.

Now the action of the coadjoint operator $\widetilde{\rm ad}^*_{(V,A)}$  on the pair of fields $(\omega,B)$
can be described as follows: for any pair of Lie algebra elements $(V,A), (W, C)\in \mathfrak g=\mathfrak X(M)\times\mathfrak X(M)$, we have

\begin{equation}
\begin{array}{rcl}
&~&\langle (W,  C),\widetilde{\rm ad}^*_{(V,A)}(\omega,B)\rangle
=\langle (W,  C),{\rm ad}^*_{(V,\sigma(A))}(\sigma(\omega),B)\rangle\\\\
&=&\langle (W,C), (L_{V}\sigma(\omega)-L_B \sigma(A),-[V,B])\rangle=\langle (W,C), (\sigma(L_{V}\omega-L_B A),-[V,B])\rangle\\\\
&=&\langle (W,C), (\sigma([\omega,V]-[A,B]),-[V,B])\rangle=\langle (W,C),([\omega,V]-[A,B],-[V,B])\rangle\,.
\end{array}
\end{equation}

The equality in the second line is due to the fact that the operator $\sigma$ commutes with the volume-preserving change of coordinates. In these computations we have used $\sigma$ to identify $\mathfrak X(M)$ and $\mathfrak X^*(M)$.

We conclude that the coadjoint action $\widetilde{\rm ad}^*_{(V,A)}$ on the vector fields  $(\omega,B)$ is
 $$
 \widetilde{\rm ad}^*_{(V,A)}(\omega,B)=([\omega,V]-[A,B],-[V,B])\,.
 $$
\end{remark}

\begin{remark}
For a general Riemannian closed three-manifold $M$, the space of divergence-free fields is the direct sum of the space of exact 
fields and the space of harmonic fields (whose dimension is equal to the first Betti number of the manifold). 
As explained in Section~\ref{sec:settings}, the magnetic helicity and the cross-helicity are defined on pairs $(\omega,B)$ 
of exact fields. Denoting the space of divergence-free exact fields on $M$ by $\mathfrak X(M)$, all the discussion in this section 
can be applied in that context with minor variations. For example, the dual space $\mathfrak X^*(M)$ 
of the Lie algebra $\mathfrak X(M)$ of exact fields is given by the space of coexact 
1-forms identified with $\Omega^1/\text{ker}(d:\Omega^1\rightarrow\Omega^2)$. In this case, the natural pairing and the operator $\sigma$ 
defined in Remark~\ref{rem:action} are well defined on exact fields because $\text{curl}^{-1}:\mathfrak X(M)\to \mathfrak X(M)$ 
is one-to-one.
\end{remark}

\end{document}